\documentclass[runningheads]{llncs}
\usepackage{amsmath}
\usepackage{graphicx}
\usepackage{color}
\usepackage{subcaption}
\usepackage{booktabs}
\usepackage{lipsum}
\usepackage{geometry}
\usepackage{multirow}
\usepackage{bm}
\usepackage{wrapfig}
\newcommand{\OurAlgo}{QOMIC }
\begin{document}
\title{QOMIC: Quantum optimization for motif identification}
%
%
\author{Hoang M. Ngo\inst{1} \and
Tamim Khatib\inst{1} \and
My T. Thai\inst{1} \and Tamer Kahveci\inst{1}}
\authorrunning{Hoang et al.}
%
\institute{University of Florida, Gainesville 32611, USA}
\maketitle              
\begin{abstract}
Network motif identification problem aims to find topological patterns in biological networks. Identifying non-overlapping motifs is a computationally challenging problem using classical computers. Quantum computers enable solving high complexity problems which do not scale using classical computers. In this paper, we develop the first quantum solution, called QOMIC (Quantum Optimization for Motif IdentifiCation), to the motif identification problem. QOMIC transforms the motif identification problem using a integer model, which serves as the foundation to develop our quantum solution. We develop and implement the quantum circuit to find motif locations in the given network using this model. Our experiments demonstrate that QOMIC outperforms the existing solutions developed for the classical computer, in term of motif counts. We also observe that QOMIC can efficiently find motifs in human regulatory networks associated with five neurodegenerative diseases: Alzheimers, Parkinsons, Huntingtons, Amyotrophic Lateral Sclerosis (ALS), and Motor Neurone Disease (MND).

\keywords{Quantum computing, motif identification, regulatory networks}
\end{abstract}
%
%
%

\section{Introduction}



Biological systems are represented as intricate networks of molecules such as genes, proteins, and metabolites interacting with each other~\cite{Zhu2007}. Uncovering potential properties of biological networks provide opportunities for gaining insights on fundamental principles that govern living organisms. In these networks, frequently recurring subgraphs are referred to as \emph{motifs}. Motifs serve as building blocks of large and complicated biological networks~\cite{Milo2002}. Studying motifs is significant as it can reveal functions of biological systems such as transcriptional regulation networks~\cite{Shen-Orr2002,Alon2007} or protein interaction networks~\cite{Wuchty2003,Ferre2015}.

Motif identification remains a computationally challenging problem as it involves solving the subgraph isomorphism, which is an NP-hard problem~\cite{Cook1971}. As the volume of biological data continues to grow rapidly, it will be more computationally expensive to explore motifs in biological networks. Enforcing specific regulatory constraints on interactions, such as activation or repression patterns on the motif topology, further increases the computation time needed to find motifs.

In literature, there are three well-known measures for counting motifs, namely $\mathcal{F}_1$, $\mathcal{F}_2$, and $\mathcal{F}_3$~\cite{Schreiber2005}. $\mathcal{F}_1$ measures counts every isomorphic subgraph of a target network to a given motif pattern with no restrictions~\cite{grochow2007network,kashtan2004efficient,omidi2009moda,wernicke2006efficient,chen2006nemofinder,kashani2009kavosh}. Although $\mathcal{F}_1$ provides a comprehensive view on all possible embeddings of motif patterns on networks, it fails to capture dependencies among motif embeddings. In addition, $\mathcal{F}_1$ does not satisfy the downward closure property~\cite{Schreiber2005}, which can lead to challenges on scaling the motif size. In contrast to $\mathcal{F}_1$, $\mathcal{F}_2$ and $\mathcal{F}_3$ measures impose restrictions on the motif embeddings. Specifically,  $\mathcal{F}_2$ does not allow resulting embeddings of given motifs in the network sharing the same edge (i.e., interaction). $\mathcal{F}_3$ restricts them from sharing the same node (i.e., molecule). Furthermore, $\mathcal{F}_2$ and $\mathcal{F}_3$ measures are downward closed~\cite{Elhesha2016}.

These distinct counting concepts provide various perspectives and trade-offs for the motif identification problem. For the $\mathcal{F}_2$ measure, the authors in~\cite{Patra2019} uses the idea of dynamic expansion trees to count motif embeddings. They also propose a new algorithm using basic building patterns to find embeddings, and then iteratively join the parent patterns with these basic building patterns, which they called pattern-join \cite{Patra2019_2}. In the work~\cite{Elhesha2016}, the authors introduced a motif-centric approach which constructs a set of basic building patterns and then explores embeddings corresponding to these patterns in term of $\mathcal{F}_2$, and $\mathcal{F}_3$. Then, in the work~\cite{ren2019}, the authors consider $\mathcal{F}_3$ for the motif identification problem in multi-layer networks. However, other than the topological constraints of $\mathcal{F}_2$ and $\mathcal{F}_3$, existing methods for the motif identification problem do not consider biological constraints within networks, such as the regulatory relations. These are additional constraints that go beyond the motif topology and enforce the type of interactions (such as activation or suppression) for each interaction in the motif. Due to the bottleneck on the current computational capacity, handling the motif identifications with multiple constraints is challenging. Therefore, there is an urgent need for a more powerful computing scheme to broaden the scope of this problem.

The field of quantum computing has drawn significant attention and investment recently because of its potential supremacy over classical computational methods~\cite{Harrow2017}. Specifically, quantum computing can address a variety of tasks which are intractable for classical computers~\cite{Shor1994,Farhi2019,Aaronson2011,Cerezo2021}. In the field of computational biology, quantum computing with its advantages shows great promise in solving complex computational biology tasks that require substantial computational resources, such as DNA alignment~\cite{Sarkar2021_a}, genome assembly~\cite{Charkiewicz2022,Boev2021}, and DNA sequence reconstruction~\cite{Sarkar2021} (see~\cite{Marchetti2022} for a short survey).

One of the most popular paradigms of quantum computing is the gate-based quantum model (a.k.a. the universal model)~\cite{Lloyd1996}. To handle combinatorial optimization problems, a quantum algorithm is introduced to work on the gate-based quantum model, named \emph{Quantum Approximate
Optimization Algorithm (QAOA)}~\cite{Farhi2014}. To use QAOA for solving a combinatorial optimization problem, several steps must be taken. First, we define an unconstrained objective function $f$ to quantify potential solutions for the combinatorial optimization problem. Then, we construct two quantum operators: a problem Hamiltonian to encode the predefined objective function $f$, and a mix Hamiltonian to expand the solution search space. Next, from two Hamiltonians, we design a parameterized quantum circuit including rotation quantum gates. This circuit operates on an initial quantum state prepared as a uniform superposition of all potential basis states. The parameters control the transformation of the initial state. In the final step, we iteratively optimize the set of parameters, using classical optimizers, such that the expectation of the state after transformation is minimized. Sampling this optimal state gives us the optimal solution to the given problem. 

In this work, we consider the motif identification problem, which finds the maximum set of motif embeddings in a target network such that these embeddings do not share any molecule (i.e., $\mathcal{F}_3$ measure) and all of these motif embeddings satisfy the regulatory constraints imposed by the given motif. We refer to our problem as the motif identification (MI) problem. We design a novel quantum solution for the MI problem, namely \emph{\OurAlgo (Quantum Optimization for Motif IdentifiCation)}. This is the first study solving this generalized MI problem using quantum computing. First, we model the MI problem as an optimization problem on the set of edges of the target network. Then, we propose an integer representation based on this model, followed by the unconstrained objective function for the MI problem. Finally, we introduce a quantum circuit design for solving the MI problem by QAOA. We implement \OurAlgo in the quantum gate-based machine provided by IBM. We compare \OurAlgo against the baseline method designed for the classical computer~\cite{ren2019} on 1500 synthetic networks and 4 motif patterns. Our results demonstrate that \OurAlgo outperforms the baseline method in terms of the motif count. Our results on real transcriptional regulatory networks for five neurodegenerative disorders suggest that \OurAlgo efficiently scales to large real networks.


\section{Preliminaries}
In this section, we first present basic concepts in quantum computing. Then, we explain the fundamental principles of QAOA, which is needed to understand our quantum computing solution to the MI problem.

\subsection{Basic concepts}
At the heart of quantum computing are \emph{quantum bits} (a.k.a. qubits), the quantum analogs of classical bits (0s and 1s). Unlike classical bits, which are either 0 or 1, a qubit can represent 0 and 1 simultaneously, exploiting the principles of quantum superposition. Information stored in a qubit is called the quantum state of that qubit, denoted by $|\psi\rangle$. Given two complex numbers $\alpha_0$ and $\alpha_1$, the quantum state of a qubit can be represented by a linear combination of two basis states $|0\rangle$ and $|1\rangle$ as: $$|\psi\rangle = \alpha_0 |0\rangle + \alpha_1 |1\rangle$$ Here, $\alpha_0$ and $\alpha_1$ represent the amplitudes associated with these basic states. 

The concept of entanglement allows us to combine multiple qubits into a quantum system, creating a quantum state that encompasses all possible combinations of the individual qubit states. For a system which includes $n$ qubits, the quantum states of the system are presented by $2^n$ basic states. For example, given four complex amplitudes $\alpha_{00}$, $\alpha_{01}$, $\alpha_{10}$, and $\alpha_{11}$, a quantum state in a 2-qubit system can be represented as: $$|\Psi\rangle = \alpha_{00} |00\rangle + \alpha_{01} |01\rangle + \alpha_{10} |10\rangle + \alpha_{11} |11\rangle$$

The quantum states of a quantum system can be transformed by \emph{quantum operators}. In the context of quantum gate-based model, these operators are referred to as \emph{quantum gates}. Common gates applied to single qubits include the Pauli-X, Y, Z gates which perform phase flips, and the Hadamard gate which creates superposition by transforming a $|0\rangle$ state into an equal superposition of $|0\rangle$ and $|1\rangle$. Quantum computing also involves gates applied to multiple qubits. One common two-qubit gate is the Controlled NOT (CNOT) gate, which is to flip the state of the second qubit if the first qubit is in the $|1\rangle$ state. Additionally, a series of quantum gates applied to qubits is called a \emph{quantum circuit}.

By a process called \emph{measurement}, we extract information from a quantum system. Measurement collapses the qubits' superposition state into a definite classical state. The outcome of a measurement is probabilistic, because it depends on the qubit's superposition amplitudes. For example, in a single qubit system with the state of $\psi$ as above, $|\alpha_0|^2$ and $|\alpha_1|^2$ are the probabilities of measuring that system as $0$ and $1$ respectively.

\subsection{Quantum Approximate Optimization Algorithm (QAOA)}

QAOA is a 4-step quantum computing paradigm designed to tackle combinatorial optimization problems.
\begin{itemize}
    \item \textbf{Step 1:} Given a combinatorial optimization problem with $n$ binary variables, we define an unconstrained objective function $f$ that quantifies the quality of the solution $\mathbf{S} \in \{0,1\}^n$ for that problem.

    \item \textbf{Step 2:} We construct two quantum operators: \emph{problem Hamiltonian}, denoted as $H_P$ and the \emph{mixing Hamiltonian}, denoted as $H_B$. Hamiltonian operators govern how the quantum state changes over time through the Schrödinger equation. Specifically, $H_P$ is used to encode the objective function of the problem. For any quantum basis state $|S\rangle$ corresponding to the solution $S \in \{0,1\}^n$, the problem Hamiltonian satisfies $H_P |S\rangle = f(S)|S\rangle$. On the other hand, the mixing Hamiltonian $H_B$ is used to perform state mixing, facilitating the exploration of solution space. Given $X_i$ as the Pauli-X gate that acts on the i\emph{th} qubit, $H_B$ can be written as $H_B = \sum_{i=1}^n X_i$.

    \item \textbf{Step 3:} Given an integer $p$, and $2p$ parameters $(\bm{\gamma}, \bm{\beta}) \equiv (\gamma_1,\dots,\gamma_p,\beta_1,\dots,\beta_p)$, along with an intial quantum state $|S_0\rangle$, we prepare a parameterized quantum circuit that transforms $|S_0\rangle$ by $2p$ operators in form of $e^{-i\gamma_jH_P}$ and $e^{-i\beta_jH_B}$ with $j \in [p]$. The final quantum state, obtained by this circuit, can be written as: $$|\bm{\gamma}, \bm{\beta}\rangle = e^{-i\beta_pH_B}e^{-i\gamma_pH_P}\dots e^{-i\beta_1H_B}e^{-i\gamma_1H_P}|S_0\rangle$$ $|\bm{\gamma}, \bm{\beta}\rangle$ is the distribution of all potential solutions for the problem, depending on parameters $(\bm{\gamma}, \bm{\beta})$.

    \item \textbf{Step 4:} We use a classic optimizers to find the optimal parameters $(\bm{\gamma^*}, \bm{\beta^*})$ such that the expectation $\langle\bm{\gamma^*}, \bm{\beta^*}|H_P|\bm{\gamma^*}, \bm{\beta^*}\rangle$ is minimum.
\end{itemize}

\section{Methods}
In this section, we first define the Motif Identification (MI) problem for biological networks (Section \ref{sec:problem definition}). We then describe the integer representation of this problem (Section~\ref{sec:boolean representation}). We present the construction of the final Hamiltonian and the design of the corresponding quantum circuit (Section~\ref{sec:quantum circuit}).
\subsection{Problem definition}\label{sec:problem definition}

Given a set of regulatory interactions among genes as $\mathcal{D}$ (e.g, $\mathcal{D} = \{"Activation", "Repression"\}$), we model a regulatory network as a connected graph $G = (V, E,\gamma)$ where $V$ is the set of nodes, $E$ is the set of edges and $\gamma: E \rightarrow \mathcal{D}$ is a mapping from an edge to its corresponding regulatory relationship. We define a motif pattern as a connected graph $M=(V', E', \gamma')$, where $V'$, $E'$ and $\gamma': E' \rightarrow \mathcal{D}$ represent the set of motif nodes, edges and the mapping from edges to regulatory relationships respectively. Given two graphs $G_1 = (V_1, E_1, \gamma_1)$ and $G_2 = (V_2, E_2, \gamma_2)$, we say that $G_1$ and $G_2$ are isomorphic if there exists a bijection (one-to-one and onto mapping) $g: V_1 \rightarrow V_2$ such that for every pair of nodes $u,v \in V_1$, we have the edge $(u,v) \in E_1$ if and only if the edge $(g(u),g(v)) \in E_2$ and the regulatory relationships between two edges are same (i.e. $\gamma_1((u,v)) = \gamma_2((g(u),g(v))$). We say that a subset of edges $\Lambda \subseteq E$, is an embedding of the motif pattern $M$ in $G$, if the induced subgraph $G[\Lambda]$ is isomorphic to $M$, denoted by $G[\Lambda] \equiv M$. We consider two embeddings $\Lambda_1$ and $\Lambda_2 \subseteq E$ to be non-overlapping if the induced subgraphs $G[\Lambda_1]$ and $G[\Lambda_2]$ do not share any nodes.

For a given set of embeddings $\mathcal{W} = \{\Lambda|\Lambda \subseteq E, G[\Lambda] \equiv M\}$, we introduce the function $\phi$, defined as $\phi(\mathcal{W}) = \cup_{\Lambda \in \mathcal{W}} \Lambda$. We refer to $\phi(\mathcal{W})$ as an \emph{edge decomposition} of $\mathcal{W}$. Additionally, we characterize a set of embeddings $\mathcal{W}$ as non-overlapping if, for all embedding $\Lambda_i$ and $\Lambda_j \in \mathcal{W}$, $\Lambda_i$ and $\Lambda_j$ are non-overlapping. From this, we formally define of the MI problem as follows:

\begin{definition}
    \textbf{(MI problem)} Consider a network $G = (V,E,\gamma)$ and a motif pattern $M = (V', E',\gamma')$. The MI problem aims to find the largest set of non-overlapping embeddings of $M$ into $G$.
    \label{definition:formal}
\end{definition}

In Definition~\ref{definition:formal}, there is no specific linkage between the given elements (the network $G$ and motif pattern $M$) and the task at hand (identifying the largest set of non-overlapping embeddings). Therefore, we examine the "non-overlapping" characteristic in term of given inputs through Lemmas~\ref{theorem:def:unique} and~\ref{theorem:def:properties}. We provide proofs of all lemmas and theorems in the Supplementary Materials.

\begin{lemma}
Consider a network $G = (V,E,\gamma)$ and a motif pattern $M = (V', E', \gamma')$. Given two sets of non-overlapping embeddings $\mathcal{W}_1$ and $\mathcal{W}_2$ of $M$ into $G$ such that $\mathcal{W}_1 \neq \mathcal{W}_2$, then $\phi(\mathcal{W}_1) \neq \phi(\mathcal{W}_2)$.
\label{theorem:def:unique}
\end{lemma}

According to Lemma~\ref{theorem:def:unique}, we deduce that, given a set of non-overlapping embeddings $\mathcal{W}$, the edge decomposition $\phi(\mathcal{W})$ is unique. Conversely, if we are given a set of edges $\mathcal{E}$ and are aware that $\mathcal{E}$ constitutes an edge decomposition of a non-overlapping embedding set $\mathcal{W}$, we are able to reconstruct $\mathcal{W}$. In Lemma~\ref{theorem:def:properties}, we establish properties that characterize a set of edges $\mathcal{E}$ as an edge decomposition. 

\begin{lemma}
Consider a network $G = (V,E,\gamma)$ and a motif pattern $M = (V', E',\gamma')$. Given an arbitrary edge set $\mathcal{E} = \{e|e \in E\}$, we show that $\mathcal{E}$ is a unique edge decomposition of a non-overlapping embedding set $\mathcal{W}$ of $M$ in $G$ if it has properties as follows:
\begin{itemize}
    \item \textbf{Property 1:} For every $e \in \mathcal{E}$, there exist a set of $|E'|-1$ distinct edges $S_e = \{\bar{e}_1, \dots, \bar{e}_{|E'|-1} \in \mathcal{E}\}$ such that $G[\{e\} \cup S_e] \equiv M$.

    \item \textbf{Property 2:} For every $e_1, e_2 \in \mathcal{E}$ such that $e_1$ and $e_2$ share a same node, then $e_1 \in S_{e_2}$ and $e_2 \in S_{e_1}$.
\end{itemize}
\label{theorem:def:properties}
\end{lemma}

From these two lemmas, we rewrite the MI problem's definition, which is equivalent to that in Definition~\ref{definition:formal}, but is better aligned to the quantum solution we will develop in the rest of this section:

\begin{definition}
    \textbf{(alternative)} Consider a network $G = (V,E,\gamma)$ and a motif pattern $M = (V', E',\gamma')$. The MI problem aims to find the largest set of edges $\mathcal{E} = \{e_i|e_i \in E\}$ such that $\mathcal{E}$ satisfies the two properties in Lemma~\ref{theorem:def:properties}.
    \label{definition:alternative}
\end{definition}

\subsection{Integer representation for the MI problem}\label{sec:boolean representation}

We model two regulatory relationships consisting of activation and repression, which we label as $0$ and $1$, respectively. Thus, the set of regulatory relationships $\mathcal{D}$ is $\{0,1\}$. Given network $G = (V, E, \gamma)$ and motif $M = (V', E', \gamma')$, we use the term $(i,j)$ with $i,j \in V$ and $(i',j')$ with $i',j' \in V'$ to denote an edge in $G$ and $M$ respectively. Additionally, given two nodes $i, j \in V$, $\gamma(i,j) = 0$ if $(i,j) \in E$ and the relationship is activation, while $\gamma(i,j) = 1$ if $(i,j) \in E$ and the relationship is repression. In case $(i,j) \notin E$, given a large constant $\Omega$, $\gamma(i,j) = \Omega$. The same rules are applied to $\gamma'$. Given two nodes $i,j \in V$ and two nodes $i', j' \in V'$, we define $c_{i,j,i',j'} = 1 - \left|\gamma(i, j) - \gamma'(i', j')\right|$. We notice that $c_{i,j,i',j'} = 1$ if $(i,j)$ and $(i',j')$ have the same regulatory relationship, or $(i,j) \notin E$ and $(i',j') \notin E'$.

We model selected edges in the solution with a set of binary variables $\mathbf{X} = \{x_{i,j}|i,j \in V\}$. More specifically, we denote edge $(i,j)$ with $x_{i,j}$ as:
\begin{equation*} \label{eq:indicator:function}
    x_{i,j} = \begin{cases}
     0 & \text{if $(i,j)$ is not selected.} \\
     1 & \text{if $(i,j)$ is selected.}
\end{cases}
\end{equation*}

Given two nodes $i,j \in V$, for $n = |V|$, we represent a permutation of indices $1,2,\dots,n$ denoting the $n$ nodes in $V$ as $[\pi_1, \pi_2, \pi_3, \dots, \pi_n]$, where $\pi_1 = i$ and $\pi_2 = j$. We define the set of all possible such permutations as $\mathcal{M}^{(n)}_{V,i,j}$.

Let us denote $n' = |V'|$. Without loss of generality, we assume that nodes in the motif $M$ are labeled from $1$ to $n'$ and there always exists an edge between node $1$ and $2$. We realize that given $(i,j) \in E$, each permutation $[\pi_1, \pi_2, \pi_3, \dots, \pi_{n'}] \in \mathcal{M}^{(n')}_{V,i,j}$ corresponds to a distinct edge set $S = \{(\pi_{i'}, \pi_{j'})| (i',j') \in E'\}$. Thus, $$\prod_{(i',j')\in E'} x_{\pi_{i'},\pi_{j'}}c_{\pi_{i'}, \pi_{j'},i',j'} = 1$$ if all edges in the set $S$ are selected and $G[S] \equiv M$ (i.e., the subgraph of $G$ induced on $S$ is isomorphic to the motif $M$). As a result, given the edge $(i,j) \in E$, the number of embeddings in $G$ including the edge $(i,j)$ is equal to the sum $$h_{V,i,j} = \sum_{[\pi_1,\dots,\pi_{n'}] \in \ \mathcal{M}^{(n')}_{V,i,j}} \prod_{(i',j')\in E'} x_{\pi_{i'},\pi_{j'}}c_{\pi_{i'}, \pi_{j'},i',j'}$$

\noindent Thus, in accordance with Definition~\ref{definition:alternative}, we formulate the MI problem as a constrained integer model as:

\textit{Maximize:}
    \[
        \sum_{(i,j) \in E} x_{i,j}
    \]

    \textit{Subject to:}
    \begin{align}
    x_{i,j} - h_{V,i,j}&= 0 &\:\:\:\:\forall (i,j) \in E \label{constraint:edge limit}\\ 
    x_{i,j} x_{k,t}(h_{V \setminus \{k,t\},i,j}+h_{V \setminus \{i,j\},k,t}) & = 0 &\:\:\:\:\forall (i,j),(k,t) \in E, |\{i,j\} \cap \{k,t\}| \geq 1 \label{constraint:node limit}\\ 
    x_{i,j} &= 0& \:\:\:\:\forall (i,j) \notin E \label{constraint:no edge}
\end{align}
The formulation above maximizes the number of selected edges such that these edges can form a set of non-overlapping embeddings. These three constraints follow the properties of non-overlapping embedding sets, as established in Lemma~\ref{theorem:def:properties}. Constraint~(\ref{constraint:edge limit}) ensures that for each selected edge $(i,j) \in E$, there exists exactly one distinct selected edge set $S_{i,j}$ such that $G[S_{i,j}] \equiv M$. This constraint corresponds to the first property in Lemma~\ref{theorem:def:properties}. Constraint~(\ref{constraint:node limit}) ensures that for every pair of selected edges $(i,j), (k,t) \in E$ which share at least one common node, there is no motif constructed by sequences of $\mathcal{M}^{(n')}_{V \setminus \{k,t\},i,j}$ and $\mathcal{M}^{(n')}_{V \setminus \{i,j\},k,t}$. This constraint corresponds to the second property in Lemma~\ref{theorem:def:properties}. Finally, Constraint~(\ref{constraint:no edge}) assigns $x_{i,j} = 0$ if the given network $G$ does not contain the edge $(i,j)$.

\begin{theorem}
    An assignment of values to the variables in $\mathbf{X}$ which maximizes the number of edges and satisfies three constraints~(\ref{constraint:edge limit}),~(\ref{constraint:node limit}) and~(\ref{constraint:no edge}) in our integer model yields the optimal solution to the MI problem.
    \label{theorem:IP correctness}
\end{theorem}

In order to design a quantum solution for the MI problem, given a set of variables $\mathbf{X}$, we represent the integer model of the MI problem as an unconstrained objective function $f: \mathbf{X} \rightarrow \mathcal{R}$. The function $f$ includes a cost function $f_c$ which evaluates the quality of the input $\mathbf{X}$ (i.e. the number of edges) and three penalty functions $f_{p_1}$, $f_{p_2}$, and $f_{p_3}$ which validate the input $\mathbf{X}$ in term of Constraints~(\ref{constraint:edge limit}), (\ref{constraint:node limit}), and~(\ref{constraint:no edge}) respectively. In details, we have:
\begin{align}
    f(\mathbf{X}) = -f_{c}(\mathbf{X}) + f_{p_1}(\mathbf{X})+ f_{p_2}(\mathbf{X})+ f_{p_3}(\mathbf{X})
\end{align}

The target function $f_{c}(\mathbf{X}) = \sum_{(i,j) \in E} x_{i,j}$ is equivalent to the target of the integer model. $f_{c}(\mathbf{X})$ returns the number of selected edges in $\mathbf{X}$. The penalty function $f_{p_1}$ ensures that the assignment $\mathbf{X}$ satisfies Constraint~(\ref{constraint:edge limit}). $f_{p_1}(\mathbf{X})$ returns 0 if $\mathbf{X}$ satisfies Constraint~(\ref{constraint:edge limit}), and returns a large number otherwise. Given a large constant $A_1$, we compute the first penalty function as:
\begin{align}
    f_{p_1}(\mathbf{X}) = A_1\sum_{(i,j) \in E}(x_{i,j} - h_{V,i,j})^2
\end{align}

Similarly, given a large constant $A_2$, we compute the second penalty function as:
\begin{align}
    f_{p_2}(\mathbf{X}) = A_2\sum_{(i,j), (k,t) \in E, |\{i,j\}\cap\{k,t\}| \geq 1} x_{i,j}x_{k,t}(h_{V \setminus \{k,t\}, i,j} + h_{V \setminus \{i,j\}, k,t})^2
\end{align}

Finally, given a large constant $A_3$, we compute the third penalty function as:
\begin{align}
    f_{p_3}(\mathbf{X}) = A_3\sum_{(i,j) \notin E}x_{i,j}
\end{align}

\begin{theorem}
    The assignment of $\mathbf{X}$, which minimizes the function $f$, optimally solves the $MI$ problem.
\end{theorem}

\subsection{A quantum circuit design for the MI problem}\label{sec:quantum circuit}
Here, we provide detailed description of the quantum circuit which QAOA employs to solve the MI problem. Given the cardinality of the set $\mathbf{X}$ as $r = |\mathbf{X}|$, the circuit is designed to operate on a $r-$qubit system. Specifically, each assignment of $\mathbf{X}$ corresponds to a basis state in the $r-$qubit system. As the mixing Hamiltonian is fixed, we need to construct the initial state $|S_0\rangle$, and the problem Hamiltonian $H_P$ for the circuit.

\begin{wrapfigure}{r}{0.5\textwidth}
  \centering
  \includegraphics[width=0.5\columnwidth]{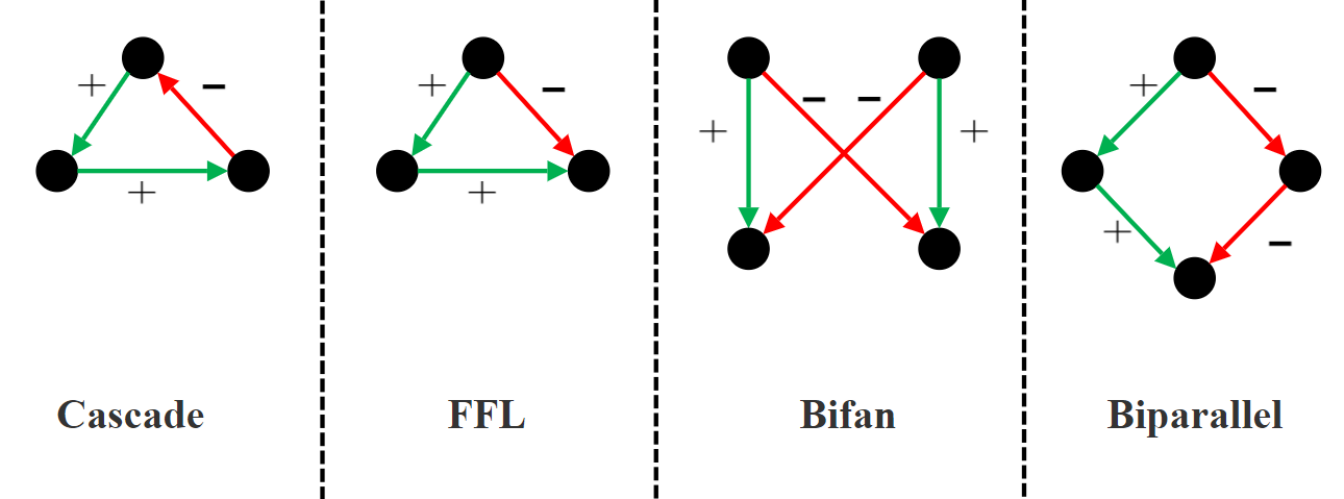}
  
    \caption{Four motif patterns with corresponding regulatory relations used in the synthetic dataset. The green color (or +) represents activation, while the red one (or -) represents repression}
    \label{fig:motif example}
\end{wrapfigure}

First, we define the initial state $|S_0\rangle$ used in QAOA as a superposition of all possible basis states with equal amplitudes. $|S_0\rangle$ can be expressed as: $$|S_0\rangle = (|0\rangle + |1\rangle)^{\otimes r}$$



We define the problem Hamiltonian $H_P$ which encodes objective function $f$ such that $H_P|\mathbf{X}\rangle = f(\mathbf{X})|\mathbf{X}\rangle$. Given a variable $x \in X$, we define $Z^{(x)}$ as the Pauli-Z gate that acts on the qubit corresponding to $x$. Given the identity operator $I$, we can construct $H_P$ by substituting each variable $x \in X$ in the objective function $f$ as $\frac{1}{2} (I-Z^{(x)})$~\cite{Hadfield2021}. By measuring this circuit, we can obtain a quantum state that represents the distribution of potential solutions for the MI problem.

\section{Experiments}
In this section, we assess the performance of \OurAlgo using synthetic and real datasets. We focus on four motif topologies, namely cascade, feed forward loop (FFL), bifan, and biparallel, which occur frequently in biological networks~\cite{Milo2002}. For each motif, we consider common regulatory relationships reported in the literature~\cite{Ingram2006,Kim2008,Papatsenko2009}. Figure~\ref{fig:motif example} depicts the motifs and their regulatory relationships.

\noindent{\bf Datasets.} We use synthetic and real datasets in our experiments.

\noindent
\emph{Synthetic datasets: } In order to examine the performance of \OurAlgo under networks with diverse topological properties, we conduct benchmarking experiments using synthetic datasets. The properties that govern synthetic datasets are as follows. First, we define the number of nodes and the average degree of nodes as $n$ and $d$, respectively. These two parameters control the size and density of the network. We define the ratio of activating interactions in a network as $r$. This parameter influences the distribution of activation and repression interactions. We construct a synthetic network by first creating a predefined number of motifs of a given motif topology. We then randomly insert edges and their regulatory interactions until the network size, density, and ratio of activation constraints are satisfied. We generate different networks by varying these parameters as: $n \in \{200, 400, 600, 800, 1000\}$, $d \in \{2, 4, 6, 8, 10\}$, $r \in \{0.2, 0.5, 0.8\}$ for each of the four motif types shown in Figure~\ref{fig:motif example}. For each combination of these parameters, we generate five synthetic networks. In total, we have $5 \times 5 \times 3 \times 4 \times 5 = 1500$ synthetic networks.

\noindent
\emph{Real datasets: } In order to evaluate the performance of \OurAlgo on real datasets, we use the Transcriptional Regulatory Relationships Unraveled by Sentence-based Text Mining (TRRUST) dataset \cite{TRRUST2017}. TRRUST is a manually curated database of human and mouse transcriptional regulatory networks, though we are only interested in the human networks. The total number of human transcriptional regulatory interactions in this dataset is 9396, with each being labeled Repression, Activation, or Unknown. We focus on neurodegenerative diseases, specifically Alzheimer's, Parkinson's, Huntington's, Amyotrophic Lateral Sclerosis (ALS), and Motor Neurone Disease (MND). Through DisGeNet, we find the Gene-Disease Associations to find which genes are related to the listed diseases, with correlation being given through a score from 0 to 1. We then tailor our TRRUST dataset to include only genes that are associated with the specific diseases being investigated.

\noindent{\bf The baseline method.}
For our baseline method, we use the method introduced in \cite{ren2019}. The method searches for all possible motif embeddings in the network, then calculates the number of embeddings that cannot be selected along with the current embedding, recording it as the loss value. It then iteratively selects embeddings with the least loss until no more independent embedding can be chosen.

\noindent{\bf Implementation.}
In the integer model for the MI problem, the number of required qubits corresponds to the number of edges in the network $G$. Due to the limitation on the number of qubits in current quantum machines, we employ a partitioning technique to address the large networks. In details, we divide the initial network $G$ into a collection of sub-networks, and then apply \OurAlgo on each sub-networks. Finally, we aggregate resulting motifs on sub-networks to derive the total motifs presented in the initial network $G$. Our partitioning technique ensures that motifs in the final solution are pairwise non-overlapping. In addition, we implement and test \OurAlgo using IBM quantum simulators~\cite{gadi_aleksandrowicz_2019_2562111}. The details of our implementation can be found in \emph{https://github.com/ngominhhoang/Quantum-Motif-Identification.git}.

\subsection{Evaluation on synthetic datasets}

\begin{figure}[t]
\begin{subfigure}{.24\textwidth}
  \centering
  \includegraphics[width=1\textwidth]{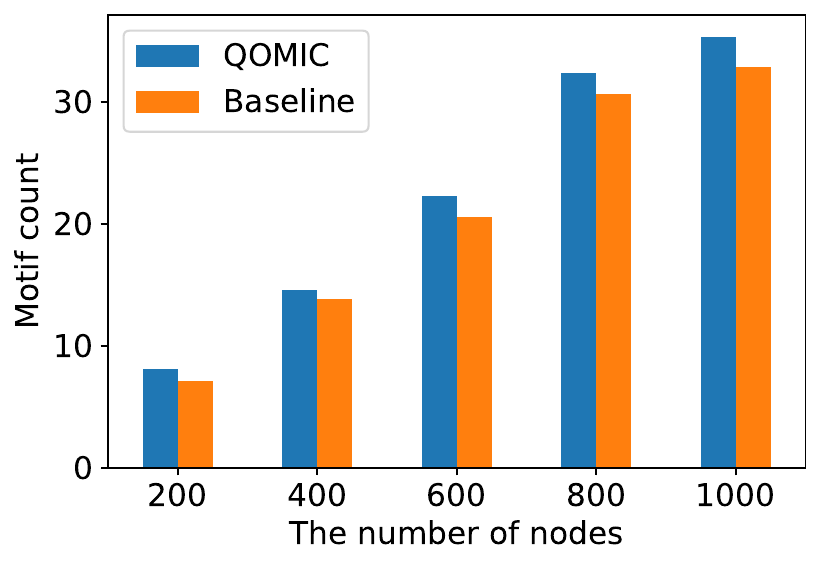}
  \caption{Cascade}
  \label{fig:syn:cascade:nums nodes}
\end{subfigure}
\hfill
\begin{subfigure}{.24\textwidth}
  \centering
  \includegraphics[width=1\columnwidth]{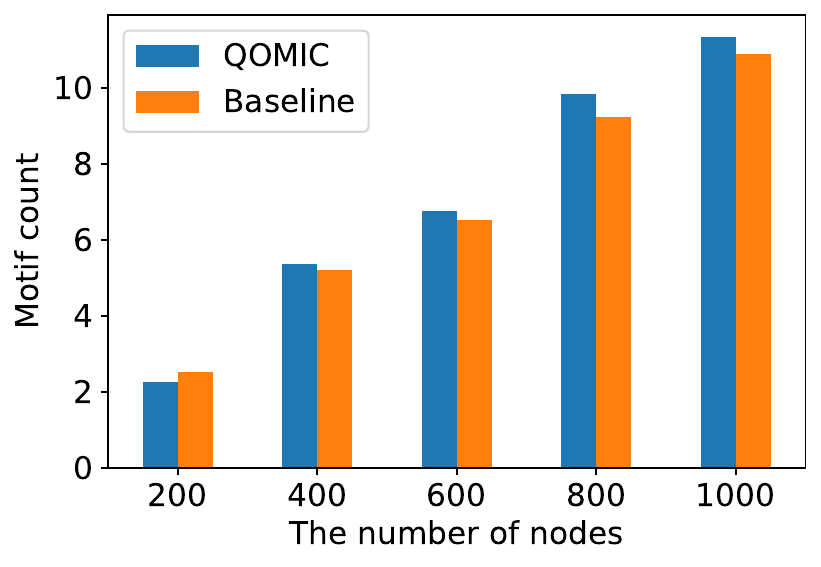}
  \caption{FFL}
  \label{fig:syn:FFL:nums nodes}
\end{subfigure}
\hfill
\begin{subfigure}{.24\textwidth}
  \centering
  \includegraphics[width=1\columnwidth]{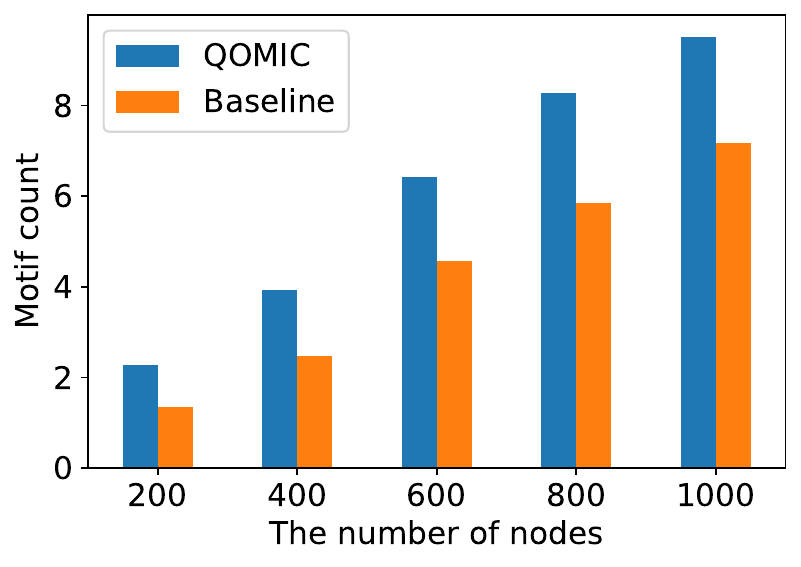}
  \caption{Bifan}
  \label{fig:syn:bifan:nums nodes}
\end{subfigure}
\hfill
\begin{subfigure}{.24\textwidth}
  \centering
  \includegraphics[width=1\columnwidth]{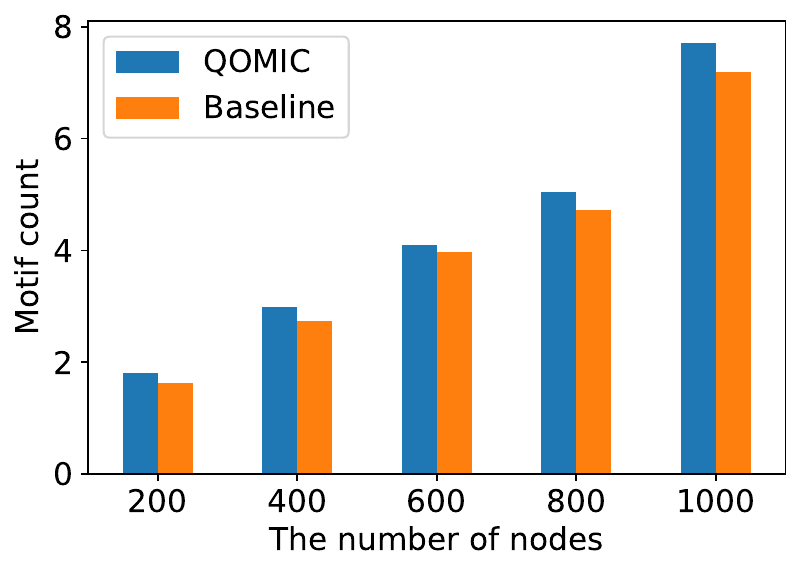}
  \caption{Biparallel}
  \label{fig:syn:biparallel:nums nodes}
\end{subfigure}
    \caption{Analysis of the \OurAlgo and the baseline method in term of the number of motif embeddings found by varying the number of nodes in the synthetic networks}
    \label{fig:syn:nums nodes}
\end{figure}

\begin{figure}[t]
\begin{subfigure}{.24\textwidth}
  \centering
  \includegraphics[width=1\textwidth]{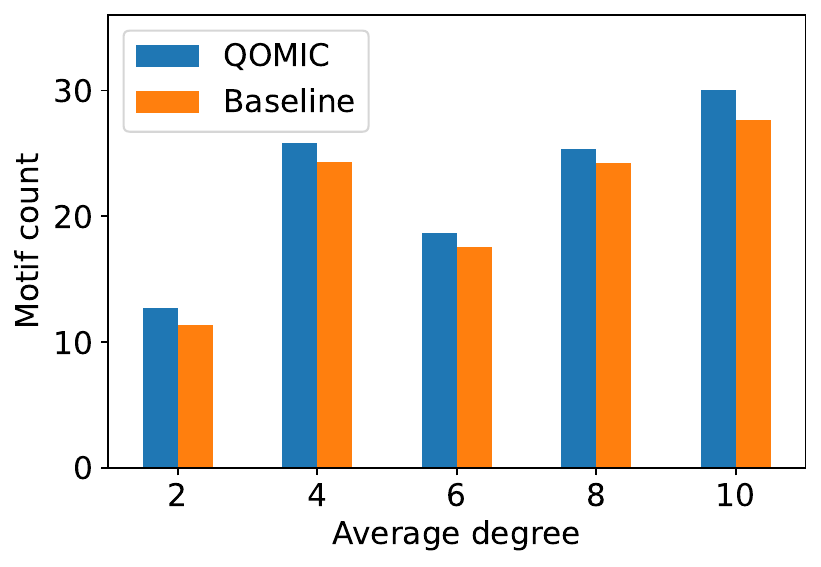}
  \caption{Cascade}
  \label{fig:syn:cascade:degree}
\end{subfigure}
\hfill
\begin{subfigure}{.24\textwidth}
  \centering
  \includegraphics[width=1\columnwidth]{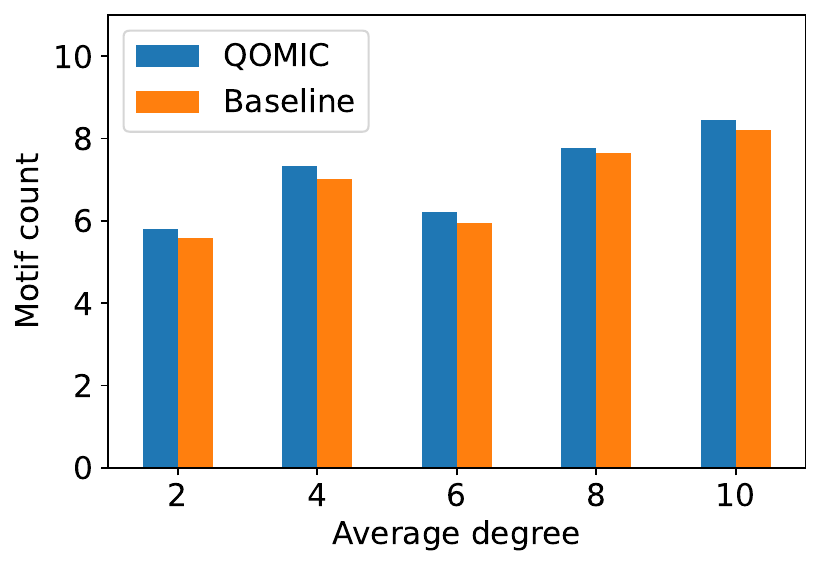}
  \caption{FFL}
  \label{fig:syn:FFL:degree}
\end{subfigure}
\hfill
\begin{subfigure}{.24\textwidth}
  \centering
  \includegraphics[width=1\columnwidth]{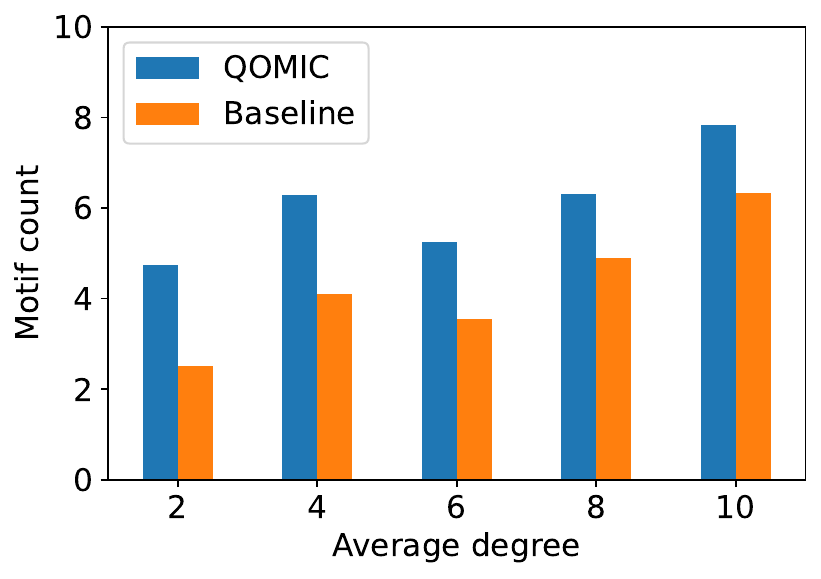}
  \caption{Bifan}
  \label{fig:syn:bifan:degree}
\end{subfigure}
\hfill
\begin{subfigure}{.24\textwidth}
  \centering
  \includegraphics[width=1\columnwidth]{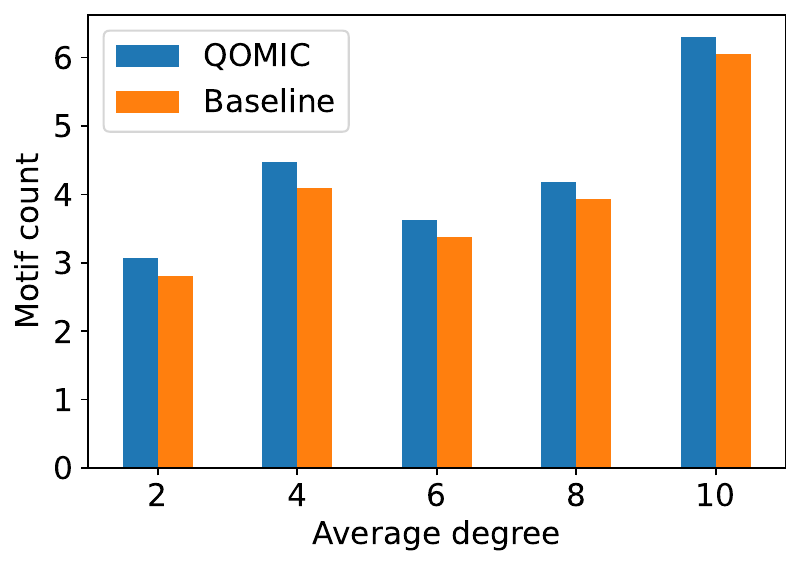}
  \caption{Biparallel}
  \label{fig:syn:biparallel:degree}
\end{subfigure}
    \caption{Analysis of the \OurAlgo and the baseline method in term of the number of motif embeddings found by varying the density in the synthetic networks}
    \label{fig:syn:degree}
\end{figure}

\begin{figure}[t]
\begin{subfigure}{.24\textwidth}
  \centering
  \includegraphics[width=1\textwidth]{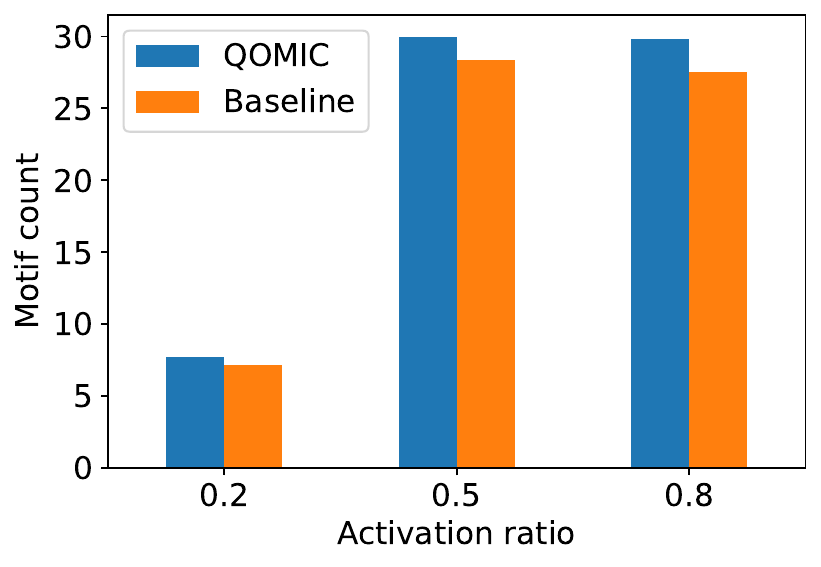}
  \caption{Cascade}
  \label{fig:syn:cascade:ratio}
\end{subfigure}
\hfill
\begin{subfigure}{.24\textwidth}
  \centering
  \includegraphics[width=1\columnwidth]{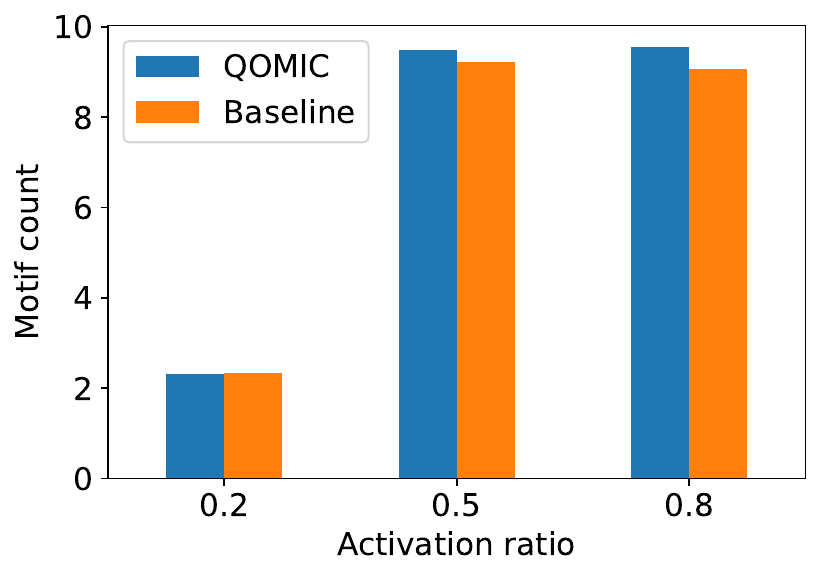}
  \caption{FFL}
  \label{fig:syn:FFL:ratio}
\end{subfigure}
\hfill
\begin{subfigure}{.24\textwidth}
  \centering
  \includegraphics[width=1\columnwidth]{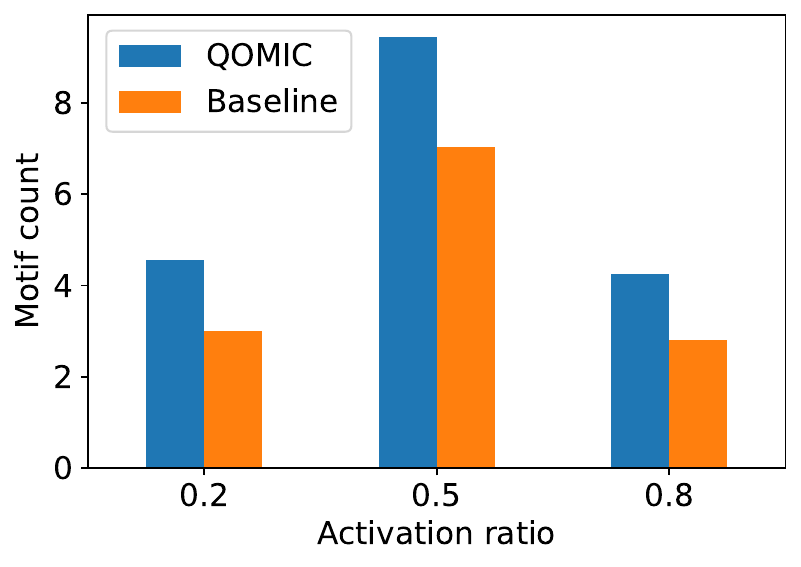}
  \caption{Bifan}
  \label{fig:syn:bifan:ratio}
\end{subfigure}
\hfill
\begin{subfigure}{.24\textwidth}
  \centering
  \includegraphics[width=1\columnwidth]{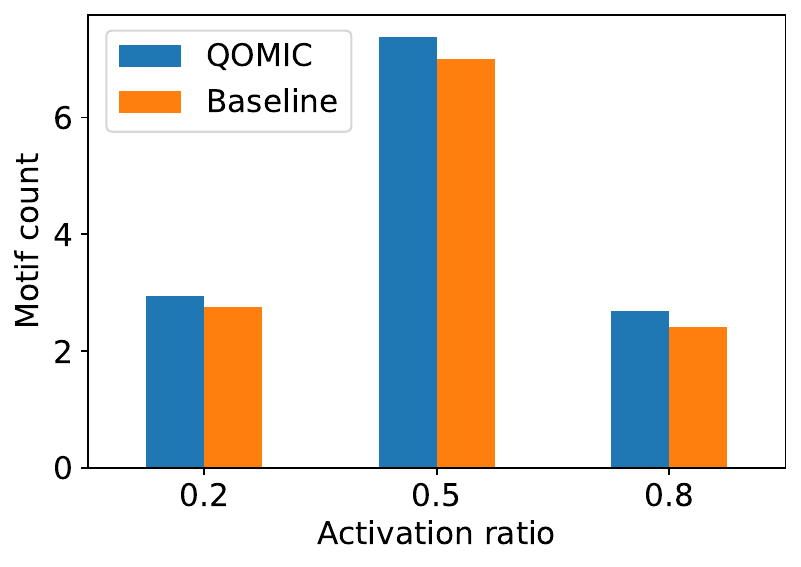}
  \caption{Biparallel}
  \label{fig:syn:biparallel:ratio}
\end{subfigure}
    \caption{Analysis of the \OurAlgo and the baseline method in term of the number of motif embeddings found by varying the activation ratio in the synthetic networks}
    \label{fig:syn:ratio}
\end{figure}

We benchmark the performance of \OurAlgo and the baseline method on different criteria of synthetic datasets including network size, network density and the distribution of regulatory relationships. For each experiment, we compare two methods in terms of the number of resulting motifs. In addition, we compare the running times of \OurAlgo and the baseline method on different network size for two motif patterns including bifan and biparallel.

\noindent{\bf The impact of network size.}
In this experiment, we compare the performance of two methods under different network sizes ranging from 200 to 1000 nodes. Figure~\ref{fig:syn:nums nodes} shows the average number of motifs resulting from \OurAlgo and the baseline method for each network size. We observe that \OurAlgo consistently outperforms the baseline method in identifying all four motif types. Specifically, the number of cascade, FFL, bifan, and biparallel motifs detected by \OurAlgo exceed those found by the baseline method by 6.2\%, 3.4\%, 41.9\%, and 6.9\%, respectively. Additionally, the disparity in solution quality between the two methods becomes more significant in the case of the bifan and biparallel motifs which possess more complex topologies compared to the cascade and FFL motifs. On the other hand, it is significant to note that cascade motifs are more prevalent than the other three motif types at the same network size. This occurrence can be attributed to the cyclic topology of the cascade motif, which relaxes constraints on the regulatory relationships within the motif. Finally, as the network size increases, the gap between \OurAlgo and the baseline method grows in favor of our method. Thus, \OurAlgo is even more advantageous when dealing with complex motif topologies and large networks.

\begin{wrapfigure}{t}{.4\textwidth}
    \vspace*{-0.5cm}
    \begin{minipage}{\linewidth}
    \centering\captionsetup[subfigure]{justification=centering}
    \includegraphics[width=\linewidth]{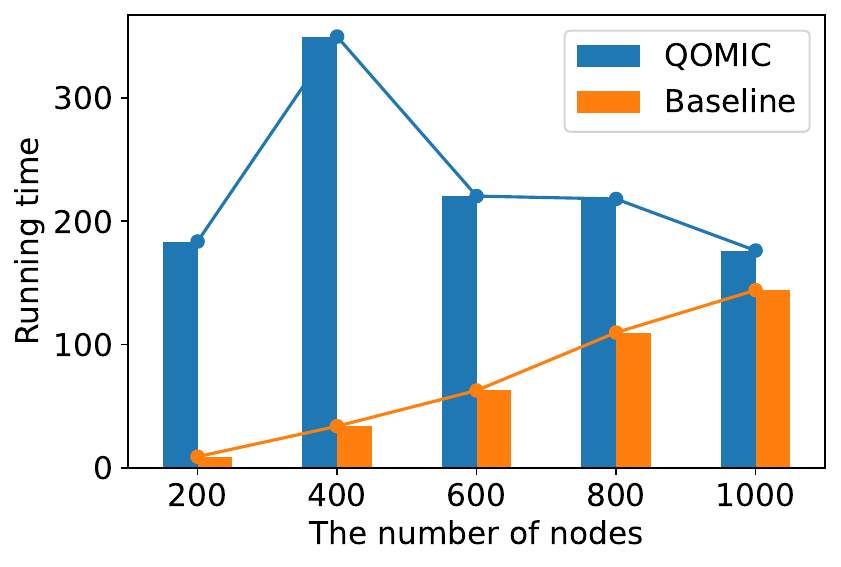}
    \subcaption{Bifan}
    \label{fig:syn:running time:bifan}\par\vfill
    \includegraphics[width=\linewidth]{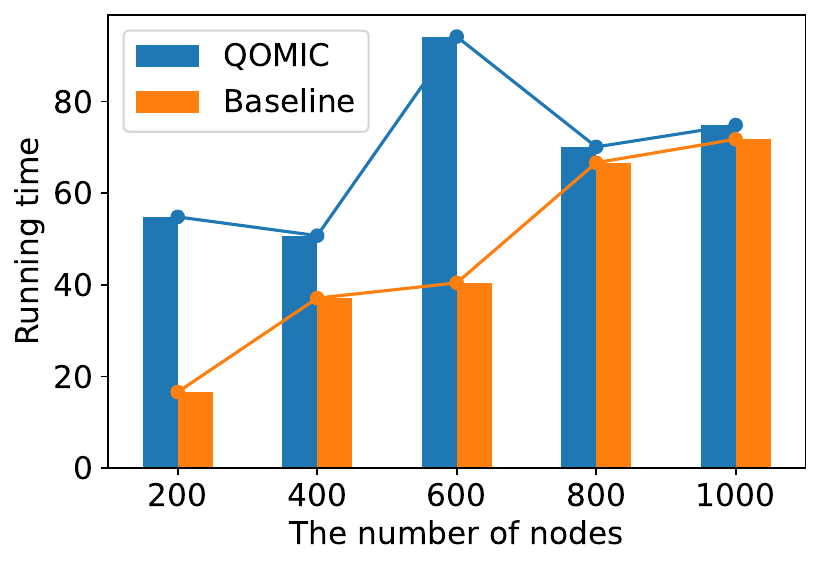}
    \subcaption{Biparallel}
    \label{fig:syn:running time:biparallel}
\end{minipage}
\caption{Analysis of the \OurAlgo and the baseline method in term of the running time}\label{fig:syn:running time}
    \vspace*{-0.5cm}
\end{wrapfigure}

\noindent{\bf The impact of network density.}
Next, we compare two methods by varying the density of networks from 2 to 10. Figure~\ref{fig:syn:degree} illustrates the average number of motifs obtained using \OurAlgo and the baseline method. Similar to the previous results, \OurAlgo provides superior solutions than the baseline method in all cases. However, unlike the previous findings where the motif count increased with the number of nodes, it is noteworthy that networks with higher average degree may yield fewer motifs. Specifically, in four cases of motif types, the number of motifs found in networks with a density of 6 is lower than in networks with a density of 4. This is due to the fact that additional edges might lead to overlapping motifs which violate the constraint of the MI problem. Consistent with our previous results, we observe more gain in motif count using our method for complex motif topologies, such as bifan.

\noindent{\bf The impact of regulatory ratio.}
Then, we consider the distribution of regulatory relationships in networks. In details, we examine the ratio of activation relationships as $0.2$, $0.5$ and $0.8$. As this ratio gets close to $0.5$, the interaction types get more heterogeneous. Figure~\ref{fig:syn:ratio} presents the average number of motifs obtained using \OurAlgo and the baseline method. We observe that \OurAlgo identifies more motifs than the baseline method in $11/12$ cases and yields the same result in $1/12$ case. Additionally, for cascade and FFL motifs, networks with activation probabilities of $0.5$ and $0.8$ exhibit a notably higher motif count compared to networks with a 0.2 ratio. On the other hand, for the bifan and biparallel motifs, the motif count of networks which have activation probability of $0.8$ surpasses the motif count of networks with two other ratios. These observations suggest that networks which have similar activation ratio with the ratio of the target motif are more likely to contain valid motif embeddings.

\noindent{\bf Running time.}
Finally, we examine the running time of \OurAlgo and the baseline method with different network sizes from 200 to 1000 nodes. The running time of \OurAlgo for a single network instance is measured as the accumulative time of four steps in QAOA applied to that instance.
Figure~\ref{fig:syn:running time} illustrates the running time comparison between two methods. We observe that while the running time of the baseline method scales linearly along with the network sizes, the running time of \OurAlgo is independent with the network sizes. It comes from the fact that the running time of \OurAlgo is heavily depended on the complexity of the quantum circuit (i.e., the total number of quantum gates to encode the objective function $f$, and the actual time to execute these gates), and the efficiency of the optimizer in finding the optimal parameters. These factors do not strictly scale with the number of network nodes. 

Furthermore, as the network size grows, the running time of \OurAlgo gets closer and even potentially surpasses the running time of the baseline method. Specifically, for graphs with 1000 nodes, the time difference between \OurAlgo and the baseline are approximately 20 seconds in for the bifan pattern, and 2 seconds for the biparallel pattern. It is important to note that the quantum computing is still an evolving field. With the rapid development in the quantum computing technology, the running time of quantum computing can be further reduced. Thus,  quantum computing is promising in solving complex biological problems with small computational cost.

\begin{wraptable}{r}{9.cm}
\vspace*{-1.5cm}
{\footnotesize
\begin{tabular}{|c|ccc||ccc||ccc||ccc|}
\hline
 & \multicolumn{3}{c||}{Cascade} & \multicolumn{3}{c||}{FFL} & \multicolumn{3}{c||}{Bifan} & \multicolumn{3}{c|}{Biparallel} \\ \cline{2-13} 
\multicolumn{1}{|l|}{} & \multicolumn{1}{l|}{MC} & \multicolumn{1}{c|}{AC} & RC & \multicolumn{1}{l|}{MC} & \multicolumn{1}{c|}{AC} & RC & \multicolumn{1}{l|}{MC} & \multicolumn{1}{c|}{AC} & RC & \multicolumn{1}{l|}{MC} & \multicolumn{1}{c|}{AC} & RC \\ \hline
Alzheimers & \multicolumn{1}{c|}{7} & \multicolumn{1}{c|}{3} & 4 & \multicolumn{1}{c|}{52} & \multicolumn{1}{c|}{43} & 37 & \multicolumn{1}{c|}{72} & \multicolumn{1}{c|}{85} & 57 & \multicolumn{1}{c|}{27} & \multicolumn{1}{c|}{27} & 25 \\ \hline
Parkinsons & \multicolumn{1}{c|}{5} & \multicolumn{1}{c|}{6} & 2 & \multicolumn{1}{c|}{48} & \multicolumn{1}{c|}{36} & 27 & \multicolumn{1}{c|}{55} & \multicolumn{1}{c|}{62} & 36 & \multicolumn{1}{c|}{22} & \multicolumn{1}{c|}{19} & 18 \\ \hline
Huntingtons & \multicolumn{1}{c|}{7} & \multicolumn{1}{c|}{4} & 5 & \multicolumn{1}{c|}{32} & \multicolumn{1}{c|}{24} & 26 & \multicolumn{1}{c|}{41} & \multicolumn{1}{c|}{51} & 30 & \multicolumn{1}{c|}{12} & \multicolumn{1}{c|}{13} & 13 \\ \hline
ALS & \multicolumn{1}{c|}{6} & \multicolumn{1}{c|}{2} & 1 & \multicolumn{1}{c|}{32} & \multicolumn{1}{c|}{17} & 24 & \multicolumn{1}{c|}{40} & \multicolumn{1}{c|}{41} & 29 & \multicolumn{1}{c|}{12} & \multicolumn{1}{c|}{15} & 8 \\ \hline
MND & \multicolumn{1}{c|}{1} & \multicolumn{1}{c|}{1} & 1 & \multicolumn{1}{c|}{3} & \multicolumn{1}{c|}{1} & 3 & \multicolumn{1}{c|}{5} & \multicolumn{1}{c|}{6} & 6 & \multicolumn{1}{c|}{1} & \multicolumn{1}{c|}{0} & 0 \\ \hline
\end{tabular}
}
\caption{The statistics on the motif count (MC), activation count (AC) and repression count (RC) per motif patterns and diseases}
\label{tab:real:statistics}
\end{wraptable}

\subsection{Evaluation on real datasets}
Here, we discuss about the efficiency of \OurAlgo in five practical human regulatory networks. These network includes genes related to neurodegenerative diseases including Alzheimers, Parkinsons, Huntingtons, ALS and MND. Motif patterns we aim to identify include cascade, FFL, bifan and biparallel. In this experiments, for the sake of generality, we identify motifs without imposing any constraint on the activation ratio.

\noindent{\bf Motif distribution.}
Table~\ref{tab:real:statistics} lists the number of motifs found, as well as the number of activation and repression relations per motif patterns and diseases. Among four motif patterns, the cascade pattern contributes the fewest number of motifs, accounting for only 5.4\% of the total, while the bifan pattern contributes the most number of motifs, with 44.3\% in total. This observation suggests that in regulatory networks associated with five diseases, the bifan topology is the most prevalent, whereas the triangle loop topology (cascade) is relatively rare. In addition, we observe that among five diseases, networks associated with Alzeheimer's and Parkinson's contribute nearly 60\% of the total number of motifs, while the network of MND only contributes roughly 2\%. This phenomenon suggests that genes related to Alzheimer's and Parkinson's exhibit strong regulatory relationships with each other by four popular motif patterns. On the other hand, among the motifs discovered, the total count of activation relations is about $1.3$ times greater than the total count of repression relations. This ratio is consistent with the ratios observed in synthetic motif patterns. Furthermore, the sum of activation and repression counts is moderately smaller than the total number of edges within the motifs found in nearly all cases. This is because of a large number of relationships between genes being categorized as unknown. On average, each motif embedding found includes approximately 50\% of unknown edges.

\noindent{\bf Frequency distribution of motif genes across diseases.}
Here, we invest in the appearance of genes in motifs found. We denote a gene associated with a motif as \emph{motif gene}. Figure~\ref{fig:real:unique genes} illustrates the number of motif genes appeared with different frequency in five diseases. In all four motifs, the number of motif genes is inversely proportional to the number of appearance of motif genes in five diseases. Specifically, the number of motif genes included in exact one disease is even more than the total number of motif genes included in more than one disease. This observation shows that each disease includes an own set of genes which are topologically related to each others. 

When we delve deeper into the proportions of these unique genes to each disease, we find out that motif genes uniquely related to Alzheimer's and Parkinson's diseases account for more than $60\%$ in total number of unique motif genes for all motif types. However, not every disease owns a strong set of unique motif genes. Specifically, the number of motif genes exclusively linked to the MND disease is fewer than the number of motif genes related to all five diseases in almost cases of motif patterns. We infer that motif genes related to the MND disease may have a broader relevance, as they are also associated with various other diseases. 

\begin{figure}[t]
\begin{subfigure}{.24\textwidth}
  \centering
  \includegraphics[width=1\textwidth]{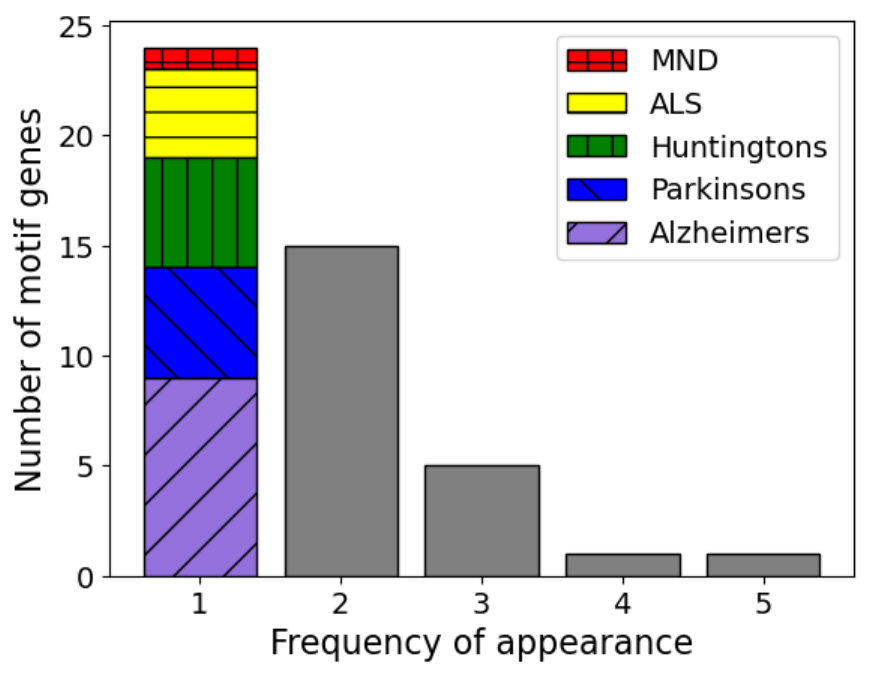}
  \caption{Cascade}
  \label{fig:real:cascade:freq}
\end{subfigure}
\hfill
\begin{subfigure}{.24\textwidth}
  \centering
  \includegraphics[width=1\columnwidth]{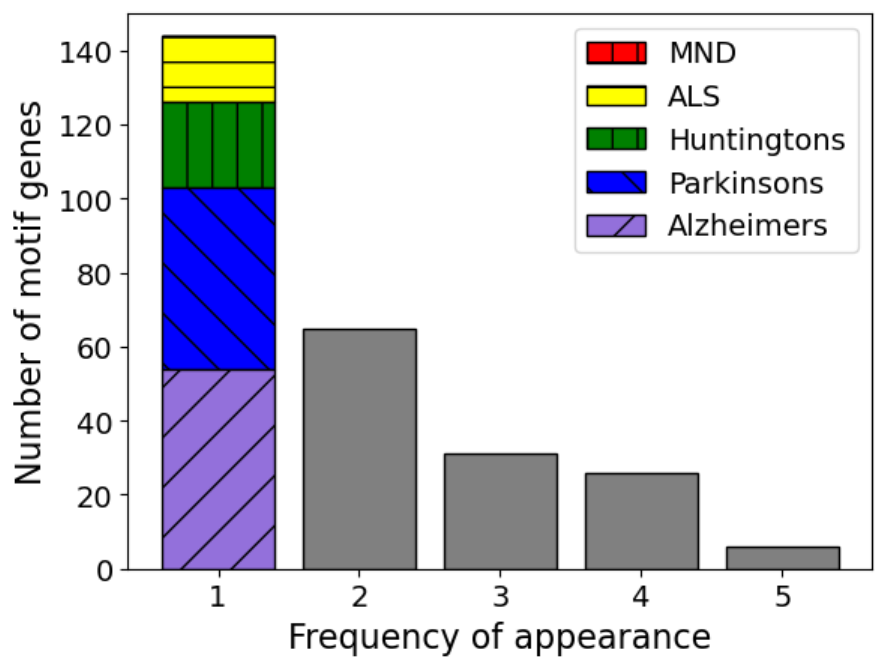}
  \caption{FFL}
  \label{fig:real:FFL:freq}
\end{subfigure}
\hfill
\begin{subfigure}{.24\textwidth}
  \centering
  \includegraphics[width=1\columnwidth]{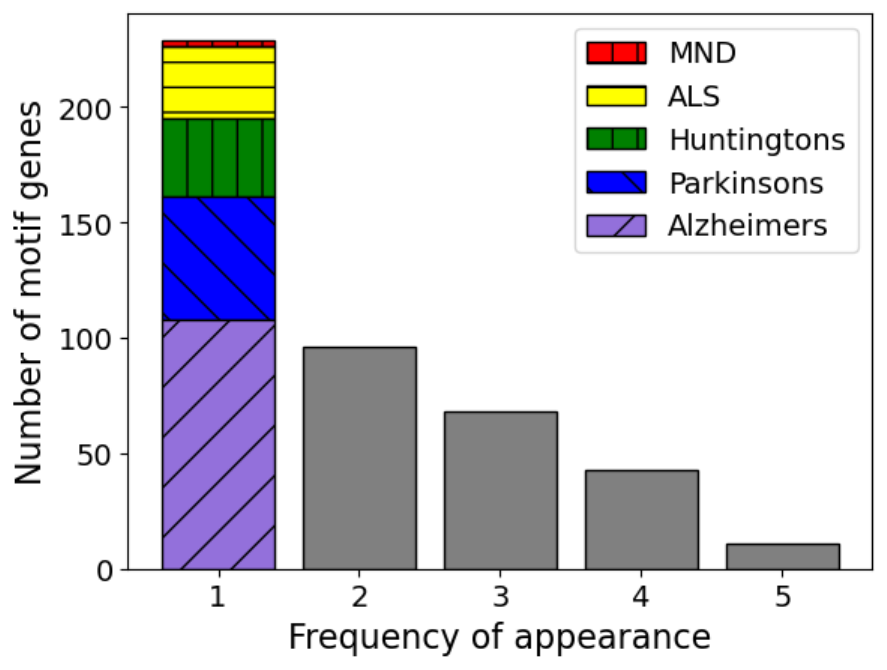}
  \caption{Bifan}
  \label{fig:real:bifan:freq}
\end{subfigure}
\hfill
\begin{subfigure}{.24\textwidth}
  \centering
  \includegraphics[width=1\columnwidth]{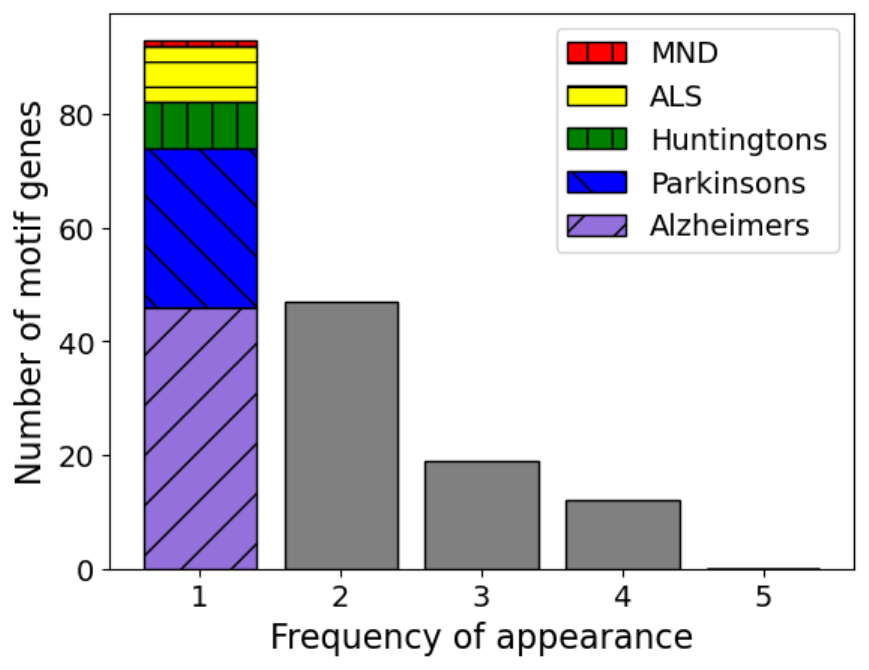}
  \caption{Biparallel}
  \label{fig:real:biparallel:freq}
\end{subfigure}
    \caption{The statistics on the number of motif genes with different appearance frequency in five diseases. The first column represents the number of genes exclusively linked to one of five diseases}
    \label{fig:real:unique genes}
\end{figure}

\noindent{\bf Statistical significance.}
Here, we examine whether the number of motif embeddings found in real datasets is significant, using the concept of z-score. In order to set up the experiment, given a real target network $G$ with the number of motif embeddings found $m_G$, we generate $N$ same-size random graphs $G_1, \dots, G_N$ by randomly shuffling the edge set of $G$. Then, we calculate the mean, denoted as $\mu$, and the standard deviation, denoted as $\sigma$, of the number of motif embeddings of N random graphs. Finally, we calculate the z-score as $Z = \frac{m_G-\mu}{\sigma}$. A motif pattern in a real network $G$ is over-represented or under-represented if the z-score is greater than $2$ or less than $-2$ respectively. 

Table~\ref{tab:statistical_significance} shows the z-scores of four motifs patterns in real networks corresponding to five diseases. We observe that the FFL and bifan patterns are over-represented in 4 and 5 over five disease-related networks respectively with high z-scores. Specifically, the z-scores of the FFL and bifan in over-representation cases are from 2 to 7.29 times higher than the threshold for being considered over-represented. It indicates that the presence of the FFL and bifan patterns in disease-related networks is statistically significant. In contrast, with the cascade and biparallel patterns, we find out that in all cases, their z-scores do not exceed 2 or fall below -2. It means that the occurrence of these two motif patterns are not significant in disease-related networks. From these results, we can infer that motif patterns, in which nodes are associated to either activation or repression, appear more frequently than expected in human regulatory networks related to diseases.

\begin{wraptable}{r}{8.5cm}
\vspace*{-2cm}
\centering
\begin{tabular}{|l||r|r|r|r|r|}
\hline
 & AD & PD & HD & ALS & MND \\ \hline \hline
Cascade & -0.26 & -0.08 & 1.62 & 0.50 & -0.05 \\ \hline
FFL & 6.93 & 8.81 & 5.90 & 5.02 & -0.01 \\ \hline
Bifan & 14.58 & 14.04 & 14.63 & 11.15 & 4.01 \\ \hline
Biparallel & -0.44 & 0.92 & -0.21 & -1.11 & -0.88 \\ \hline
\end{tabular}
\caption{The z-scores that represent statistical significance of the motif count of 4 motif patterns in 5 real regulatory networks associated with neurodegenerative diseases (AD = Alzheimer's, PD = Parkinson's, HD = Huntington Disease).}
    \vspace*{-0.5cm}
\label{tab:statistical_significance}
\end{wraptable}

\section{Conclusion}
Network motif identification problem is a significant problem in the field of biology, especially when we incorporate the $\mathcal{F}_3$ measure for the target network and introduce regulatory constraints to the motif pattern. However, the limited computational capacity of classical computers become a hindrance to the scalability of traditional methods to this problem. In this work, using quantum computing scheme, we propose a novel quantum solution, named \OurAlgo, for the MI problem. We implement and test the performance of \OurAlgo on the IBM's quantum gate-based machine. Although quantum computing is still in the early states of development, the experimental results on both synthetic and real datasets show that \OurAlgo can efficiently identify motifs within reasonable running times. In terms of motif count, \OurAlgo even outperforms the baseline method in almost all cases. This suggests that quantum computing is a promising approach in solving complex biological problems in the future.



%
%
%
%
\newpage
\bibliographystyle{splncs04}
\bibliography{bibliography}

\newpage

\setcounter{table}{0}
\setcounter{figure}{0}
\setcounter{section}{0}
\setcounter{theorem}{0}
\setcounter{lemma}{0}
\renewcommand{\thetable}{SM. \arabic{table}}
\renewcommand{\thefigure}{SM. \arabic{figure}}

\onecolumn

\pagebreak
\pagestyle{empty}
\section*{{\Large Supplementary Materials 1}}
\section{SM 1: Correctness of the QOMIC algorithm}

\begin{lemma}
    Consider a network $G = (V,E)$ and a motif pattern $M = (V', E')$. Given two sets of non-overlapping embeddings $\mathcal{W}_1$ and $\mathcal{W}_2$ such that $\mathcal{W}_1 \neq \mathcal{W}_2$, then $\phi(\mathcal{W}_1) \neq \phi(\mathcal{W}_2)$.
\label{theorem:supplementary:unique}
\end{lemma}

\begin{proof}
We prove this lemma by contradiction. We assume that there exists two different sets of non-overlapping embeddings $\mathcal{W}_1$ and $\mathcal{W}_2$ such that $\phi(\mathcal{W}_1) = \phi(\mathcal{W}_2)$. Without loss of generality, we denote an embedding $\Lambda = \{e_1,\dots,e_{|V'|}\}$ such that $\Lambda \in \mathcal{W}_1$ and $\Lambda \notin \mathcal{W}_2$. Recall that we have $\phi(\mathcal{W}_1) = \phi(\mathcal{W}_2)$, so it follows that $e_1,\dots,e_{|V'|}$ must belong to at least two different embeddings of $\mathcal{W}_2$ \textbf{(1)}.

Because we assume that the motif $M$ is connected and have at least 3 nodes, the embedding $\Lambda$ is itself a connected component with at least 3 nodes \textbf{(2)}.

From \textbf{(1)} and \textbf{(2)}, we can infer that there exists two edges in $\Lambda$ that share a same node and belong to two different embeddings in $\mathcal{W}_2$. This violates the assumption of $\mathcal{W}_2$ which is supposed to contain all non-overlapping embeddings. Thus, given two sets of non-overlapping embeddings $\mathcal{W}_1$ and $\mathcal{W}_2$ such that $\mathcal{W}_1 \neq \mathcal{W}_2$, we can prove that $\phi(\mathcal{W}_1) \neq \phi(\mathcal{W}_2)$.
\end{proof}

\begin{lemma}
Consider a network $G = (V,E)$ and a motif pattern $M = (V', E')$. Given an arbitrary edge set $\mathcal{E} = \{e|e \in E\}$, we show that $\mathcal{E}$ is a unique edge decomposition of a non-overlapping embedding set $\mathcal{W}$ of $M$ into $G$ if it has properties as follows:
\begin{itemize}
    \item \textbf{Property 1:} For every $e \in \mathcal{E}$, there exist a set of $|E'|-1$ distinct edges $S_e = \{e_1, \dots, e_{|E'|-1} \in \mathcal{E}\}$ such that $G[\{e\} \cup S_e] \equiv M$.

    \item \textbf{Property 2:} For every $e_1, e_2 \in \mathcal{E}$ such that $e_1$ and $e_2$ share a same node, then $e_1 \in S_{e_2}$ and $e_2 \in S_{e_1}$.
\end{itemize}
\label{theorem:supplementary:properties}
\end{lemma}

\begin{proof}
$\mathcal{E}$ is a unique edge decomposition of a non-overlapping embedding set $\mathcal{W}$ of $M$ into $G$ if $\mathcal{E}$ satisfies two following conditions:
\begin{itemize}
    \item \textbf{Condition 1:} There exists a unique way to completely assign every edge in $\mathcal{E}$ into distinct groups $\Lambda_1, \dots, \Lambda_m$ with $m = \frac{|\mathcal{E}|}{|E'|}$ such that $G[\Lambda_i] \equiv M$ with $i = 1, \dots ,m$.

    \item \textbf{Condition 2: } $\Lambda_1, \dots, \Lambda_m$ are pairwisely non-overlapping.
\end{itemize}
We will prove that an edge set $\mathcal{E}$ with Property 1 and 2 can satisfy Condition 1 and 2 above.

With Property 1, we can establish a method for assigning edges in $\mathcal{E}$ to $m$ distinct groups, as outlined in the first condition. To begin, we select an arbitrary edge, denoted as $e^{(1)} \in \mathcal{E}$. Then, we select a edge set $S_{e^{(1)}} = \{e^{(1)}_1, \dots, e^{(1)}_{|E'|-1}\}$ such that $G[\{e^{(1)}\} \cup S_{e^{(1)}}] \equiv M$. The existence of $S_{e^{(1)}}$ is guaranteed by Property 1. Then, we assign $\{e^{(1)}\} \cup S_{e^{(1)}}$ as the first group $\Lambda_1$. 

Moving forward, we select another arbitrary edge $e^{(2)} \in \mathcal{E} \setminus \Lambda_1$. Similarly, we choose a edge set $S_{e^{(2)}} = \{e^{(2)}_1, \dots, e^{(2)}_{|E'|-1}\}$ such that $G[\{e^{(2)}\} \cup S_{e^{(2)}}] \equiv M$. We assign $\{e^{(2)}\} \cup S_{e^{(2)}}$ as the second group $\Lambda_2$. Because of Property 1, we can prove that $\Lambda_1$ and $\Lambda_2$ are distinct. In other words, we demonstrate that $e^{(2)}_1, \dots, e^{(2)}_{|E'|-1} \in \mathcal{E} \setminus \Lambda_1$. By contradiction, we assume that there exists an edge $e' \in S_{e^{(2)}}$ such that $e' \in \Lambda_1$. As a result, for $e'$, there exists two different $S_{e'}$ such that $G[\{e'\} \cup S_{e'}] \equiv M$ that contradicts to Property 1. Thus, $\Lambda_1$ and $\Lambda_2$ are distinct.

By following a similar approach, given $i-1$ groups, we can form the i\emph{th} group $\Lambda_i = \{e^{(i)}\} \cup S_{e^{(i)}}$. Here, $e_i$ is chosen such that $e_i \in \mathcal{E} \setminus \cup_{j = 1}^{i-1} \Lambda_j$ while $S_{e^{(i)}} = \{e_1^{(i)}, \dots, e_{|E'|-1}^{(i)}\}$ satisfies $G[\{e^{(i)}\} \cup S_{e^{(i)}}] \equiv M$. Group $\Lambda_i$ is distinct with $i-1$ previous groups. In the end, we can construct $m$ groups that are the embeddings of $M$ into $G$ and pairwisely distinct \textbf{(1a)}.

Next, we show that the set of $m$ groups $\mathcal{W} = \{\Lambda_1, \dots, \Lambda_m\}$ constructed as above are unique. By contradiction, we assume that there exists a different set of $m$ distinct groups $\mathcal{W'} = \{\Lambda'_1, \dots, \Lambda'_m\}$ such that $\cup_{i=1}^m \Lambda'_i = \mathcal{E}$ and $G[\Lambda_i] \equiv \mathcal{E}$ for $i = 1,\dots, m$. Additionally, because $\mathcal{W} \neq \mathcal{W'}$, there exists at least one group $\Bar{\Lambda} \in \mathcal{W'}$ such that $\Bar{\Lambda} \notin \mathcal{W}$. Given $\Bar{e} \in \Bar{\Lambda}$, because $\Bar{e} \in \mathcal{E} = \cup_{\Lambda \in \mathcal{W}} \Lambda$, there exists a group $\Lambda \in \mathcal{W}$ such that $\Bar{e} \in \Lambda$. That contradicts to Property 1. Thus, the set of $m$ groups $\mathcal{W} = \{\Lambda_1, \dots, \Lambda_m\}$ is unique \textbf{(1b)}.

From \textbf{(1a)} and \textbf{(1b)}, we prove that if the edge set $\mathcal{E}$ has Property 1, it can satisfy Condition 1 \textbf{(1)}.

On the other hand, Property 2 implies that there is no two groups that share a same node. Thus, Condition 2 holds \textbf{(2)}.

From \textbf{(1)} and \textbf{(2)}, we prove the correctness of this lemma.
\end{proof}

\begin{theorem}
    An assignment of $\mathbf{X}$ which maximizes the number of edges and satisfies three Constraints~(\ref{constraint:edge limit}),~(\ref{constraint:node limit}) and~(\ref{constraint:no edge}) results in the optimal solution for the MI problem.
    \label{theorem:supplementary:IP correctness}
\end{theorem}
\begin{proof}
 We recall that given $i, j \in V$, each sequence $[\pi_1 = i, \pi_2 = j, \pi_3, \dots, \pi_k]$ corresponds to a distinct edge set $S = \{(\pi_{i'}, \pi_{j'})| (i',j') \in E'\}$. Thus, $$\prod_{(i',j')\in E'} x_{\pi_{i'},\pi_{j'}}c_{\pi_{i'}, \pi_{j'},i',j'} = 1$$ if all edges in the set $S$ are selected and $G[S] \equiv M$. As a result, given the edge $(i,j) \in E$, the sum $$h_{V,i,j} = \sum_{[\pi_1,\dots,\pi_{n'}] \in \ \mathcal{M}^{(n')}_{V,i,j}} \prod_{(i',j')\in E'} x_{\pi_{i'},\pi_{j'}}c_{\pi_{i'}, \pi_{j'},i',j'}$$, is equal to the number of motifs including $(i,j)$.

 Constraint~(\ref{constraint:edge limit}) is satisfied if $\forall (i,j) \in E$, $x_{i,j} = h_{V,i,j}$. If the edge $(i,j)$ is not selected with $x_{i,j} = 0$, Constraint~(\ref{constraint:edge limit}) always holds because all products in $h_{V,i,j}$ includes $x_{i,j}$. On the other hand, if the edge $(i,j)$ is selected  with $x_{i,j} = 1$, the number of motifs including $(i,j)$ must be 1. Thus, selected edges that satisfy Constraint~(\ref{constraint:edge limit}) also satisfy Property 1 in Lemma~\ref{theorem:def:properties} $\textbf{(1)}$.

 When it comes to Constraint~(\ref{constraint:node limit}), it is always satisfied if the edge $(i,j) \in E$ or $(k,t) \in E$ is not selected. On the other hand, if there exists two selected edges $(i,j), (k,t) \in E$ which share at least one common node, Constraint~(\ref{constraint:node limit}) is satisfied if $h_{V \setminus \{k,t\},i,j}+h_{V \setminus \{i,j\},k,t} = 0$. In other word, there exists no motif which includes the edge $(i,j)$, but does not include the edge $(k,t)$, and vice versa. Thus, combining with Constraint~(\ref{constraint:node limit}) in which each selected edge must associate with one motif, we can imply that $(i,j)$ and $(k,t)$ must belong to a same motif in order to satisfy two these constraints. Thus, selected edges that satisfy Constraint~(\ref{constraint:edge limit}) and (\ref{constraint:node limit}) also satisfy Property 2 in Lemma~\ref{theorem:def:properties} $\textbf{(2)}$.

 Constraint~(\ref{constraint:no edge}) is to ensure that $x_{i,j} = 0\forall (i,j) \notin E$ $\textbf{(3)}$.

 From $\textbf{(1)}$, $\textbf{(2)}$, $\textbf{(3)}$, and Lemma~\ref{theorem:def:properties}, a feasible assignment $\mathbf{X}$, that satisfies three constraints, corresponds to a valid edge decomposition of a non-overlapping embedding set of $M$ into $G$. Thus, by finding the maximum feasible $\mathbf{X}$, we can obtain the maximum number of non-overlapping motifs.

\end{proof}

\begin{theorem}
    The assignment of $\mathbf{X}$, which minimizes the function $f$, optimally solve the $MI$ problem.
\end{theorem}

\begin{proof}
    From the definition of the function $f$, a feasible $\mathbf{X}$, which satisfies three constraints in the integer model, lead to minimum values of 0 in penalty terms including $f_{p_1}$, $f_{p_2}$, and $f_{p_3}$. Besides, $f_c$ corresponds to the number of selected edges from $\mathbf{X}$ with a negative sign. Consequently, $\mathbf{X}$, which minimizes the function $f$, represents a maximum feasible solution for the integer model. From Theorem~\ref{theorem:IP correctness}, we can conclude that $\mathbf{X}$, which minimizes the function $f$, is the optimal solution for the $MI$ problem.
\end{proof}

\section{SM 2: The enrichment analysis corresponding to five neurodegenerative disorders}


Here, we represent the enrichment analysis of unique motif genes associated with the Alzheimer's disease. Our enrichment analysis is conducted using the GO software~\cite{Ashburner2000,Aleksander2003,Thomas2022}.

Table~\ref{tab:enrichment:Alzheimers} represents the top three molecular functions with the lowest false detection rates (FDR) corresponding to sets of motif genes which are uniquely associated with the Alzheimer's disease. We observe that the FDRs of all molecular functions are small (less than $10^{-6}$), so sets of unique motif genes are strongly relevant to the molecular functions found. Specifically, unique motif genes of three-node patterns, including cascade and FFL, are relevant to DNA-binding functions. On the other hand, unique motif genes of four-node patterns, such as bifan and biparallel, are associated with transcription activities. 

Table~\ref{tab:enrichment:Parkinsons} represents the top three molecular functions with the lowest FDRs corresponding to sets of motif genes associated with Parkinson's disease. Interestingly, the FDR for the cascade motif pattern is significantly higher (less than $10^{-2}$), the FDR for the FFL motif is significantly lower (less than $10^{-10}$), while bifan and biparallel remain in between these extremes (less than $10^{-6}$ and less than $10^{-5}$, respectively). 

For Huntington's disease (Table~\ref{tab:enrichment:Huntingtons}), only one molecular function corresponding to FFL was found, while the other motif patterns had at least three molecular functions. These FDRs are noticeably higher than Alzheimer's and Parkinson's, with all but one being less than $10^{-2}$. Bifan in particular struggled, with all FDRs less than $10^{-1}$.

In case of the ALS disease (Table~\ref{tab:enrichment:ALS}), only the FFL motif pattern has at least 3 molecular functions corresponding to ALS. Bifan only had one and the other motif patterns had none. The FDRs for the top three molecular functions of FFL are less than $10^{-2}$, and the FDR for the sole bifan molecular function is less than $10^{-1}$.

For the MND disease (Table~\ref{tab:enrichment:MND}), only the biparallel motif pattern has corresponding molecular functions in MND. However, its top three FDRs are incredibly low, with all being less than $10^{-5}$.

\begin{table}[t]
\centering
\begin{tabular}{|c|c|l|c|}
\hline
Motif & Term ID & \multicolumn{1}{c|}{Term description} & FDR \\ \hline
\multirow{3}{*}{Cascade} & GO:0140297 & DNA-binding transcription factor binding & 6.19e-09 \\ \cline{2-4} 
 & GO:0061629 & RNA polymerase II-specific DNA-binding transcription factor binding & 3.67e-08 \\ \cline{2-4} 
 & GO:0000978 & RNA polymerase II cis-regulatory region sequence-specific DNA binding & 1.99e-06 \\ \hline
\multirow{3}{*}{FFL} & GO:0140110 & Transcription regulator activity & 2.45e-10 \\ \cline{2-4} 
 & GO:0000978 & RNA polymerase II cis-regulatory region sequence-specific DNA binding & 3.60e-10 \\ \cline{2-4} 
 & GO:0000977 & RNA polymerase II transcription regulatory region sequence-specific DNA binding & 4.78e-10 \\ \hline
\multirow{3}{*}{Bifan} & GO:0000976 & Transcription cis-regulatory region binding & 4.89e-11 \\ \cline{2-4} 
 & GO:1990837 & Sequence-specific double-stranded DNA binding & 4.89e-11 \\ \cline{2-4} 
 & GO:0140110 & Transcription regulator activity & 5.13e-11 \\ \hline
\multicolumn{1}{|l|}{\multirow{3}{*}{Biparallel}} & \multicolumn{1}{l|}{GO:0140110} & Transcription regulator activity & \multicolumn{1}{l|}{2.61e-11} \\ \cline{2-4} 
\multicolumn{1}{|l|}{} & \multicolumn{1}{l|}{GO:0008134} & Transcription factor binding & \multicolumn{1}{l|}{2.72e-11} \\ \cline{2-4} 
\multicolumn{1}{|l|}{} & \multicolumn{1}{l|}{GO:0140297} & DNA-binding transcription factor binding & \multicolumn{1}{l|}{3.64e-10} \\ \hline
\end{tabular}
\caption{The enrichment analysis in term of molecular functions corresponding to motif genes from the Alzheimers-related network
}
\label{tab:enrichment:Alzheimers}
\end{table}

\begin{table}[t]
\centering

\begin{tabular}{|c|l|l|l|}
\hline
Motif & \multicolumn{1}{c|}{Term ID} & \multicolumn{1}{c|}{Term description} & \multicolumn{1}{c|}{FDR} \\ \hline
\multirow{3}{*}{Cascade} & GO:0000977 & RNA polymerase II transcription regulatory region sequence-specific DNA binding & 0.008 \\ \cline{2-4} 
 & GO:0000981 & DNA-binding transcription factor activity, RNA polymerase II-specific & 0.008 \\ \cline{2-4} 
 & GO:0000978 & RNA polymerase II cis-regulatory region sequence-specific DNA binding & 0.0292 \\ \hline
\multirow{3}{*}{FFL} & GO:0000976 & Transcription cis-regulatory region binding & 3.81e-12 \\ \cline{2-4} 
 & GO:0000978 & RNA polymerase II cis-regulatory region sequence-specific DNA binding & 9.63e-12 \\ \cline{2-4} 
 & \multicolumn{1}{c|}{GO:0000977} & RNA polymerase II transcription regulatory region sequence-specific DNA binding & 1.7e-11 \\ \hline
\multirow{3}{*}{Bifan} & GO:0000977 & RNA polymerase II transcription regulatory region sequence-specific DNA binding & 4.6e-07 \\ \cline{2-4} 
 & GO:0000978 & RNA polymerase II cis-regulatory region sequence-specific DNA binding & 4.6e-07 \\ \cline{2-4} 
 & GO:0003690 & Double-stranded DNA binding & 4.6e-07 \\ \hline
\multicolumn{1}{|l|}{\multirow{3}{*}{Biparallel}} & GO:0043565 & Sequence-specific DNA binding & 5.42e-06 \\ \cline{2-4} 
\multicolumn{1}{|l|}{} & GO:0000978 & RNA polymerase II cis-regulatory region sequence-specific DNA binding & 6.28e-06 \\ \cline{2-4} 
\multicolumn{1}{|l|}{} & GO:1990837 & Sequence-specific double-stranded DNA binding & 7.03e-06 \\ \hline
\end{tabular}
\caption{The enrichment analysis in term of molecular functions corresponding to motif genes from the Parkinsons-related network
}
\label{tab:enrichment:Parkinsons}
\end{table}

\begin{table}[t]
\centering

\begin{tabular}{|c|l|l|l|}
\hline
Motif & \multicolumn{1}{c|}{Term ID} & \multicolumn{1}{c|}{Term description} & FDR \\ \hline
\multirow{3}{*}{Cascade} & GO:1990841 & Promoter-specific chromatin binding & 3e-06 \\ \cline{2-4} 
 & GO:0019899 & Enzyme binding & 0.0131 \\ \cline{2-4} 
 & GO:0019901 & Protein kinase binding & 0.0131 \\ \hline
FFL & GO:0005515 & Protein binding & 0.005 \\ \hline
\multirow{3}{*}{Bifan} & GO:0042802 & Identical protein binding & 0.0275 \\ \cline{2-4} 
 & GO:0140297 & DNA-binding transcription factor binding & 0.0419 \\ \cline{2-4} 
 & GO:0005102 & Signaling receptor binding & 0.0459 \\ \hline
\multicolumn{1}{|l|}{\multirow{3}{*}{Biparallel}} & GO:0000978 & RNA polymerase II cis-regulatory region sequence-specific DNA binding & 0.0024 \\ \cline{2-4} 
\multicolumn{1}{|l|}{} & GO:0000981 & DNA-binding transcription factor activity, RNA polymerase II-specific & 0.0024 \\ \cline{2-4} 
\multicolumn{1}{|l|}{} & GO:0140110 & Transcription regulator activity & 0.0024 \\ \hline
\end{tabular}
\caption{The enrichment analysis in term of molecular functions corresponding to motif genes from the Huntingtons-related network
}
\label{tab:enrichment:Huntingtons}
\end{table}

\begin{table}[t]
\centering

\begin{tabular}{|c|l|l|l|}
\hline
Motif & \multicolumn{1}{c|}{Term ID} & \multicolumn{1}{c|}{Term description} & \multicolumn{1}{c|}{FDR} \\ \hline
\multirow{3}{*}{FFL} & GO:0000987 & Cis-regulatory region sequence-specific DNA binding & 0.002 \\ \cline{2-4} 
 & GO:0003682 & Chromatin binding & 0.002 \\ \cline{2-4} 
 & GO:0008134 & Transcription factor binding & 0.002 \\ \hline
Bifan & GO:0005178 & Integrin binding & 0.0263 \\ \hline
\end{tabular}
\caption{The enrichment analysis in term of molecular functions corresponding to motif genes from the ALS-related network
}
\label{tab:enrichment:ALS}
\end{table}

\begin{table}[t]
\centering

\begin{tabular}{|c|l|l|l|}
\hline
Motif & \multicolumn{1}{c|}{Term ID} & \multicolumn{1}{c|}{Term description} & \multicolumn{1}{c|}{FDR} \\ \hline
\multirow{3}{*}{Biparallel} & GO:0046332 & SMAD binding & 3.46e-11 \\ \cline{2-4} 
 & GO:0070411 & I-SMAD binding & 5.33e-10 \\ \cline{2-4} 
 & GO:0005160 & Transforming growth factor beta receptor binding & 1.75e-6 \\ \hline
\end{tabular}
\caption{The enrichment analysis in term of molecular functions corresponding to motif genes from the MND-related network
}
\label{tab:enrichment:MND}
\end{table}

\end{document}